\documentclass{llncs}
\usepackage{amsmath,amssymb,latexsym,epsfig,llncsdoc,graphicx,wrapfig,subfigure,ifthen}
\usepackage{amsmath}
\usepackage{amsfonts}
\usepackage{amssymb}
\usepackage{amssymb}
\usepackage{setspace,cite,multirow}

\newtheorem{lem}{Lemma}

\newtheorem{thm}{Theorem}

\date{}

\title{Cutting a Convex Polyhedron Out of a Sphere %\\(Extended Abstract)
\thanks{An earlier version appeared in Proc. WALCOM 2010, LNCS, Springer, 2010.}}
\author{Syed Ishtiaque Ahmed, Masud Hasan, and Md. Ariful Islam}

\institute{
Department of Computer Science and Engineering\\
Bangladesh University of Engineering and Technology\\
Dhaka-1000, Bangladesh\\
\texttt{http://www.buet.ac.bd/cse}\\ 
\email{\texttt{ishtiaque@csebuet.org, masudhasan@cse.buet.ac.bd, arifulislam@csebuet.org}}
}

\begin{document}

\maketitle{}

\begin{abstract}
Given a convex polyhedron $P$ of $n$ vertices inside a sphere $Q$, we give an $O(n^3)$-time algorithm that cuts $P$ out of $Q$
by using guillotine cuts and has cutting cost $O(\log^2 n)$ times the optimal.

\smallskip

\end{abstract}

\section{Introduction} 

The problem of cutting a convex polygon $P$ out of a piece of planar material $Q$ 
($P$ is already drawn on $Q$)
with minimum total cutting length is a well studied problem in computational geometry.
The problem was first introduced by Overmars and Welzl in 1985~\cite{aa}
but has been extensively studied in the last decades~\cite{AIH,aa,ab,ac,ae,ag,ah,ai,am,JK03}
with several variations, such as $P$ and $Q$ are
convex or non-convex polygons, $Q$ is a circle, and the cuts are line cuts or ray cuts.
%The results include: indication to the hardness of optimality,
%several $O(\log n)$ and constant factor approximation algorithms and a PTAS.
%See~\cite{AIH} for a summary of these results.   
%\subsection{Known results}
This type of cutting problems have many industrial applications such as in metal sheet cutting, paper cutting,
furniture manufacturing, ceramic industries, fabrication, ornaments, and leather industries.
Some of their variations also fall under \emph{stock cutting problems}~\cite{ab}.

If $Q$ is another convex polygon with $m$ edges, this problem with line cuts 
has been approached in various ways~\cite{aa,ab,ac,ad,ae,af,ah,ai}.
%by many researchers in computational geometry community~\cite{aa,ab,ac,ad,ae,af,ah,ai}.
%But if $Q$ is a circle, the problem is not approached by anyone to our knowledge
%Overmars and Welzl first introduced this problem in 1985~\cite{aa}.
%of cutting a convex polygon out of another convex polygon optimally with line cuts~\cite{aa}.
If the cuts are allowed only along the edges of $P$,
Overmars and Welzl~\cite{aa} proposed an $O(n^{3}+m)$-time algorithm for this problem with optimal cutting length,
where $n$ is the number of edges of $P$.
%However they proved that such cutting sequence may not always be optimal.
The problem is more difficult if the cuts are more general, i.e.,
they are not restricted to touch only the edges of $P$.
In that case, Bhadury and Chandrasekaran showed that the problem has optimal solutions
that lie in the algebraic extension of the input data field~\cite{ab},
and due to this algebraic nature of this problem,
an approximation scheme
%\footnote{A \emph{$\rho$-approximation algorithm}
%(similarly, an \emph{approximation scheme}) has a cutting length that is
%$\rho$ times (similarly, $(1+\epsilon)$ times, for any value $\epsilon>0$)
%the optimal cutting length. Please refer to~\cite{CLRS01} for preliminaries
%on approximation algorithms.}
is the best that one can achieve~\cite{ab}.
They also gave an approximation scheme
with pseudo-polynomial running time~\cite{ab}.
%, that is, the running time of their algorithm is polynomial only if the input data is encoded in unary.

After the indication of Bhadury and Chandrasekaran~\cite{ab} to
the hardness of the problem, many people have given polynomial time approximation algorithms.
Dumitrescu proposed an $O(\log n)$-approximation algorithm with
$O(mn+n\log n)$ running time~\cite{ae,ad}. Then, Daescu and
Luo~\cite{af} gave the first constant factor approximation
algorithm with ratio $2.5+||Q||/||P||$,
where $||P||$ and $||Q||$ are the perimeters of $P$ and the minimum area bounding rectangle of $P$ respectively.
%Note that the value of $||Q||/||P||$ can be arbitrarily large in worst case.
Their algorithm has a running time of
$O(n^{3}+(n+m)\log{(n+m)})$. The best known constant factor
approximation algorithm is due to Tan~\cite{ac} with an approximation ratio of
$7.9$ and running time of $O(n^{3}+m)$. In the same
paper~\cite{ac}, the author also proposed an $O(\log n)$-approximation
algorithm with improved running time of $O(n+m)$.
%They also gave an approximation scheme [2] with pseudo-polynomial running time.
%, that is, the running time of their algorithm is polynomial only if the input data is encoded in unary.
As the best known result so far, very recently, Bereg, Daescu and
Jiang~\cite{ah} gave a polynomial time approximation scheme (PTAS)
%(i.e., the approximation ratio is $(1+\epsilon)$ for any $\epsilon>0$)
for this problem with running time $O(m+\frac{n^{6}}{\epsilon^{12}})$.
Recently, Ahmed et.al.~\cite{AIH} have given similar constant factor and $O(\log n)$-factor
approximation algorithms where $Q$ is a circle. 
As observed in~\cite{AIH}, algorithms for $Q$ being a convex polygon are not
easily transferred for $Q$ being a circle, as the running time of the
formers depend upon the number of edges of $Q$.

%This problem has also been studied when $Q$ is a non-convex polygon~\cite{aa}
%and $P$ is a non-convex polygon~\cite{ai,af,ac} and/or the cuts are ray cuts~\cite{ai,af,ac}.

For ray cuts, Demaine, Demaine and Kaplan~\cite{ag} gave a linear time algorithm to decide
whether a given polygon $P$ is \emph{ray-cuttable} or not.
For optimally cutting $P$ out of $Q$ by ray cuts, 
if $Q$ is convex and $P$ is non-convex but ray-cuttable, 
then Daescu and Luo~\cite{af} gave an almost linear time
$O(\log^2n)$-approximation algorithm.
If $P$ is convex, then they gave a linear time $18$-approximation algorithm.
Tan~\cite{ac} improved the approximation ratio for both cases
as $O(\log n)$ and $6$, respectively, but with much higher running time of $O(n^3+m)$.
See Table~\ref{fi:comparison} for a summary of these results.

\begin{table*}[htbp]
\begin{center}
{\scriptsize
\begin{tabular}{|c|c|c|c|c|c|c|}
\hline
Dim. & Cut Type & $Q$ & $P$ & Approx. Ratio & Running Time & Reference \\
%\hline
%\multicolumn{6}{c}{2D} \\
\hline
\multirow{11}{*}{2D} 
& \multirow{5}{*}{Line} & Convex & Convex & $O(\log n)$ & $O(mn+n\log n)$  & \cite{ae,ad} \\
\cline{3-7}
& & Convex & Convex & $2.5+||Q||/||P||$ & $O(n^{3}+(n+m)\log{(n+m)})$ & \cite{af} \\
\cline{3-7}
& & Convex & Convex & 7.9 & $O(n^{3}+m)$ & \cite{ac} \\
\cline{3-7}
& & Convex & Convex & $(1+\epsilon)$ & $O(m+\frac{n^{6}}{\epsilon^{12}})$ & \cite{ah} \\
\cline{3-7}
& & Circle & Convex & $O(\log n)$ & $O(n)$ & \cite{AIH} \\
\cline{3-7}
& & Circle & Convex & $6.48$ & $O(n^3)$ & \cite{AIH} \\
\cline{2-7}
& \multirow{5}{*}{Ray} & - & Non-convex & Ray-cuttable? & $O(n)$ & \cite{ag} \\
\cline{3-7}
& & Convex & Convex & $18$ & $O(n)$ & \cite{af} \\
\cline{3-7}
& & Convex & Non-convex & $O(\log^2n)$ & $O(n)$ & \cite{af} \\
\cline{3-7}
& & Convex & Convex & $6$ & $O(n^3+m)$ & \cite{ac} \\
\cline{3-7}
& & Convex & Non-convex & $O(\log n)$ & $O(n^3+m)$ & \cite{ac} \\
\hline
%\multicolumn{6}{c}{{\bf 3D}}\\
%\hline
\multirow{2}{*}{{\bf 3D}} & Hot-wire & - & Non-convex & Cuttable? & $O(n^5)$ & \cite{JK03}\\
\cline{2-7}
& {\bf Guillotine} & {\bf Sphere} & {\bf Convex} & {\boldmath$O(\log^2 n)$} & {\boldmath$O(n^3)$} & {\bf This paper} \\
\hline
%\caption{Comparison of the results.}
%\label{fi:comparison}
\end{tabular}
}\end{center}
\caption{Comparison of the results.}
\label{fi:comparison}
\end{table*}

\paragraph{Our results}
The generalization of this problem in 3D is very little known.
To the best of our knowledge, the only result is to decide whether a 
polyhedral object can be cut out form a larger block using continuous hot wire cuts~\cite{JK03}.
In this paper we attempt to generalize the problem in 3D.
We consider the problem of cutting a convex polyhedron $P$ which is 
fixed inside a sphere $Q$ by using only guillotine cuts with minimum total cutting cost.
A \emph{guillotine cut}, or simply a \emph{cut}, is a plane that does not pass through $P$ 
and partitions $Q$ into two smaller convex pieces.
After a cut is applied, $Q$ is updated to the piece that contains $P$.
The \emph{cutting cost} of a guillotine cut is the area of the newly created face of $Q$.
We give an $O(n^3)$-time algorithm that cuts $P$ out of $Q$ by using only guillotine cuts and
has cutting cost no more than $O(\log^2 n)$ times the optimal cutting cost.
Also see Table~\ref{fi:comparison}.

The rest of the paper is organized as follows.
We give some preliminaries in Section~\ref{se:pre}, 
then Section~\ref{se:algo} gives the algorithms, and finally Section~\ref{se:con} concludes the paper 
with some future work.

\section{Preliminaries}
\label{se:pre}
%Before we start describing these two phases, we need some more preliminaries.
A cut is a \emph{vertex/edge/face cut} if it is tangent to $P$ at a single vertex/a single edge/a face respectively.
%\footnote{``tangent to a vertex'' means not tangent to an edge or a face and ``tangent to an edge'' means not tangent to a face.}
We call $P$ to be \emph{cornered} (within $Q$) if it does not contain the center $o$ of $Q$, otherwise it is called \emph{centered}.
For cornered $P$, the \emph{D-separation} of $P$ is the minimum-cost (single) cut that separates $P$ from $o$.

We represent an \emph{orthogonal projection} of a convex polyhedron $P$ by the corresponding 
\emph{projection direction} coming towards the origin from a \emph{view point} at infinity. 
A face $f$ of $P$ is visible in an orthogonal projection if the view point
lie in the half space that is defined by the supporting plane of $f$ and does not contain $P$.
An orthogonal projection of $P$ is called  \emph{non-degenerate}
%w.r.t some visible faces of $P$ if none of those faces is parallel to the projection direction.
if the projection direction is not parallel to any face of $P$.
An orthogonal projection of $P$ is a convex polygon.
If the projection is non-degenerate, then each edge of the projected convex polygon corresponds to an edge of $P$. 
%On the other hand, if the projection is degenerate, then the faces of the polyhedron
%that are parallel to the projection direction becomes an edge in the projected convex polygon.  

\section{The algorithm} 
\label{se:algo}
%The overall idea of our algorithm is as follows.
Let $C^*$ be the optimal cutting cost.
We shall have two phases in our algorithm: \emph{box cutting phase} and \emph{carving phase}.
In the box cutting phase, we shall cut a minimum volume rectangular box $B$ containing $P$ out of $Q$
with cutting cost no more than a constant factor of $C^*$.
Then in the carving phase we shall cut $P$ out of $B$ with cutting cost bounded by $O(\log ^2n)$ times of $C^*$.
%In the latter phase we shall heavily use orthogonal projection of $P$.

%In the next two subsections we describe box cutting phase and carving phase of our algorithm.

\subsection{Box cutting phase} 
We first deal with cornered $P$.
If $P$ is cornered, we shall first apply the D-separation to $Q$.
The following lemma gives a characterization of the D-separation, which will help finding it quickly.
%and based on this characterization we shall then find it quickly.

\begin{lemma}
\label{le:closest_x}
For cornered $P$, let $x$ be the closest point of $P$ from $o$.
Then the D-separation of $P$ is the plane that is perpendicular to the line segment $\overline{ox}$ at $x$.
\end{lemma}

\begin{proof}
A D-separation is a tangent to $P$ that separates $P$ from $o$ and is farthest form $o$.
Let $\pi$ be the plane that is perpendicular to the line segment $\overline{ox}$ at $x$.
To prove that $\pi$ is the D-separation, we first prove that $\pi$ is tangent to $P$.
Suppose not. 
Then there exists some portion of $P$ in the neighborhood of $x$ that lies in the half space of $\pi$ containing $o$.
Then there must be a point $y$ in that portion that is closer to $o$ than $x$,
%Let $y$ be a point in that portion.
%Then $y$ is closer to $o$ than $x$, 
which is a contradiction that $x$ is closest to $o$.

We next prove that $\pi$ is the farthest tangent of $P$ from $o$ that separates $x$ from $o$.
%$\pi$ must separate $x$ from $o$. So, it must intersect $\overline{ox}$.
To separate $x$ from $o$, $\pi$ must intersect $\overline{ox}$.
Now, any other plane that intersects $\overline{ox}$ and is not perpendicular to $\overline{ox}$ at $x$
is closer to $o$ than $\pi$.
Therefore, $\pi$ is the farthest.
\qed
\end{proof}

Observe that since $P$ is convex, the  closest point $x$ of $P$ from $o$ is unique,
and therefore, the D-separation of $P$ is also unique.
However, $x$ can be a vertex, or a point of an edge or a face.

\begin{lem}
The D-separation can be found in $O(n)$ time.
\end{lem}

\begin{proof}
By Lemma~\ref{le:closest_x}, we need to find the closest point $x$ of $P$ from $o$.
We first find the closest vertex $v$ from $o$ in $O(n)$ time.
Then for each edge $e$, we find the closest point $o_e$ of $e$ from $o$ as follows:
Let $l_e$ be the line passing through $e$.
Draw a line segment $\overline{oo'}$ perpendicular to $l_e$.
If $o'$ is a point of $e$, then $o_e$ is $o'$,
otherwise $o_e$ is the end point of $e$ that is closer to $o'$.
Finding $o_e$ can be done in constant time. 
For all edges of $P$, it takes $O(n)$ time.
Similarly, for each face $f$, we find the closest point $o_f$ of $f$ from $o$ as follows:
Let $\pi_f$ be the supporting plane of $f$.
Draw a line segment $\overline{oo'}$ perpendicular to $\pi_f$.
If $o'$ is a point of $f$, then $o_f$ is $o'$,
otherwise $o_f$ is among the closest point of the edges of $f$ or among the vertices of $f$.
Finding $o_f$ can be found in $O(d_f)$ time, where $d_f$ is the number of edges of $f$.
For all faces of $P$, it takes $\sum_fO(d_f)=O(n)$ time.
Finally, $x$ is the closest among $v$, all $o_e$'s and $o_f$'s.
\qed
\end{proof}

\iffalse*****************************************

By Lemma~\ref{le:closest_x} we need to find the plane that is perpendicular to the line segment $\overline{ox}$
at $x$, where $x$ is the closest point of $P$ from $o$.
To check whether $x$ is a vertex of $P$, 
among all the vertices of $P$ let $v$ be closest to $o$.
%for each vertex $v$ we draw a plane $\pi_v$
Let $\pi_v$ be the plane 
perpendicular to $\overline{ov}$ at $v$. If $\pi_v$ is tangent to $P$, 
then $\pi_v$ is the D-separation.
Checking $\pi_v$ to be a tangent of $P$ can be easily done in $O(d_v)\in O(n)$ time, 
where $d_v$ is the degree of $v$.
%, by finding whther every other vertex of $P$ lies in the half space of $\pi_v$ that does not contain $o$.
%Over all $v$, the time required is  $O(n)$.

To check whether $x$ is a point of an edge $e$ (similarly, of a face $f$) of $P$,
for each edge $e$ (each face $f$) we draw the line segment $\overline{oo'}$
perpendicular to the line $l_e$ passing through $e$ (perpendicular to the supporting plane $\pi_f$ of $f$).
If $o'$ is a point of $l_e$ ($\pi_f$), then $x$ is $o'$
and the D-separation is the plane perpendicular to $\overline{ox}$ at $x$.
Checking whether $o'$ to be a point of $f$ can be easily done in $O(d_f)$ time, 
where $d_f$ is the number of edges of $f$.
Over all faces of $P$, this checking requires a total time of $\sum_fO(d_f)=O(n)$.
\qed
\end{proof}

*********************************************************\fi

For cornered $P$, after the D-separation is applied, $Q$ is a spherical segment
and let $r$ be the radius of the base circle of that segment.
The following lemma relates for cornered $P$ the cost of D-separation and $C^*$.

\begin{lem}
\label{le:cornered_lb}
For cornered $P$, cost of the D-separation, which is $\pi r^2$, is at most $C^*$.
\end{lem}

\begin{proof}
Consider an optimal cutting sequence ${\cal C}$ with cutting cost $C^*$. 
${\cal C}$ must separate $P$ from $o$. 
However, it may do that by using a single cut  or by using more than one cut. 
If it uses a single cut, then it is in fact doing the the D-separation, 
since the D-separation is the minimum cost single cut that can separate $P$ from $o$.
Therefore, $C^*\ge \pi r^2$.

If ${\cal C}$ uses more than one cut, then let ${\cal C}=C_1,C_2,\ldots,C_k,\ldots$ 
with $C_k$ being the first cut that separates $o$ from $P$. 
Observe that $C_k$ can not be the very first cut of ${\cal C}$, 
since otherwise it is doing a D-separation and we are in the previous case.
Replace ${\cal C}=C_1,C_2,\ldots,C_k$ by a single cut ${C_k}'$ whose plane is the same as that of $C_k$.
%Observe that $C_k'$ 
We will show that cost of ${C_k}'$ is smaller than the total cost of $C_1,C_2,\ldots,C_{k}$.

Consider the first two cuts $C_1$ and $C_2$.
Replace $C_1$ and $C_2$ by a single cut ${C_2}'$ whose plane is the same as that of $C_2$.
Since $C_1$ does not separate $P$ from $o$, the portion of $Q$ that is created due to $C_1$
and that does not contain $P$ is no larger than a half sphere of $Q$.
It implies that the portion of ${C_2}'$ that is not present in $C_2$ is smaller than $C_1$
(also see Fig.~\ref{fi:replace}).
That means the cost of ${C_2}'$ is smaller than the total cost of $C_1$ and $C_2$.
Similarly, we can show that replacing ${C_2}'$ and $C_3$ with a 
single cut ${C_3}'$ in the plane of $C_3$ has smaller cost than the total cost of ${C_2}'$ and $C_3$.
Repeating this for $k-1$ times would show that ${C_k}'$ has smaller cost than the total cost of $C_1,C_2,\ldots,C_{k}$.
Therefore, using more than one cut to separate $P$ from $o$ is even worse than using a single cut,
and we already proved that an optimal way to use a single cut is to use the D-separation.
Thus the lemma holds.
\qed
\end{proof}

\iffalse

%The proof depends upon the fact that the cuts in an optimal cutting sequence must be tangents to $P$.
%Overmars and Welzl~\cite{aa} proved this fact for 2D, whose 3D generalization also holds.
%The argument is that if $c$ is the first cut that is not a tangent of $P$, then the cost of $c$ and the subsequent cuts behaves, 
%while moving $c$ parallelly, as a concave function in the distance of $c$ from $P$.
%Therefore, the minimum cost is achieved when it touches $P$ or is infinitely away from $P$.

The above fact also implies that,
%With the above fact, the authors in~\cite{AIH} proved in 2D that to separate $P$ from $o$ 
%an optimal cutting sequence must use the D-separation.
%The 3D generalization of this fact also holds.
%The main idea is that, 
to separate $o$ from $P$ if a single cut is used that is not a D-separation,
then it must have cost more than the D-separation, since D-separation is the minimum such cut.
If more than one cut are used, then their total cost would be even higher.
\qed
\end{proof}

\fi

\begin{figure}
\begin{center}
\input{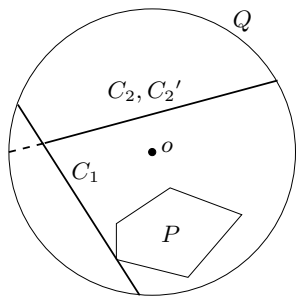}
\caption{2D view of $C_1,C_2$ and $C_2'$. 
Broken line represents the part of $C_2'$ that is not in $C_2$;
This portion is smaller than $C_1$.}
\label{fi:replace}
\end{center}
\end{figure}

We now deal with centered $P$. A lemma that is similar to the previous one and gives lower bound
for centered $P$ is the following.

\begin{lem}
\label{le:centered_lb}
For centered $P$, it holds that $C^*\ge \pi R^2$, where $R$ is the radius of $Q$.
\end{lem}

\begin{proof}
For centered $P$, $Q$ remains a sphere.
Since $P$ contains the center $o$ of $Q$, any cutting sequence, starting from the boundary of $Q$,
must ``wrap'' $P$ and finally get out of $Q$ by a plane different from the starting plane.
That means the wrapping must enclose the center $o$.
In the best case when $P$ is simply a point that lies in the center $o$ of $P$, 
the cutting sequence must traverse at least $\frac{1}{2}\pi R^2$ area to reach $P$ 
and then to traverse another $\frac{1}{2}\pi R^2$ area to finish the cutting. 
In the worst case, when $P$ is almost the sphere $Q$, the sequence must traverse the whole area of $Q$,
which is $4\pi R^2$.
\qed
\end{proof}

We next find a minimum volume rectangular bounding box $B$ of $P$ in $O(n^3)$ time 
by the algorithm of O'Rourke~\cite{O85}.
Then we cut out this box from $Q$ by applying six cuts along the six faces of $B$.

\begin{lem}
\label{le:garbage1}
Cost of cutting $B$ out of $Q$ is at most $3C^*$ for cornered $P$ and at most $4C^*$ for centered $P$.
\end{lem}

\begin{proof}
%Let $S$ be the surface of $Q$. % after applying the D-separation, if necessary.
Consider $Q$ before $B$ was cut out of it.
We denote the area of $Q$ by $|Q|$\footnote{In the subsequent text, we use $|\cdot|$ to denote the area of a 3D object.}.
For cornered $P$, since $Q$ is no bigger than a half sphere, it holds that $|Q| \le 3\pi r^2$,
which by Lemma~\ref{le:cornered_lb} becomes $|Q| \le 3C^*$. 
For centered $P$, since we do not apply D-separation, we have $|Q|=4\pi R^2$,
which by Lemma~\ref{le:centered_lb} becomes $|Q|\le 4C^*$.
While cutting along the faces of $B$, for each cut $c$ let $Q'$ be the portion of $Q$
that does not contain $P$. 
%We call $Q'$ the \emph{garbage} piece of $Q$.
Let $q'$ be the portion of the surface of $Q'$ that is ``inherited'' from $Q$, 
i.e., that was a part of the surface of $Q$.
One important observation is that the cost of $c$ is no more than the area of $q'$.
Moreover, over all six cuts, sum of these inherited surface area is $|Q|$.
Therefore, the lemma holds.
\qed
\end{proof}

Once the minimum area bounding box $B$ has been cut, a lower bound on $C^*$ can be 
given in terms of the area of $B$.

\begin{lem}
\label{le:lb}
$C^*\ge \frac{1}{6}|B|$, where $|B|$ is the area of $B$.
\end{lem}

\begin{proof}
Let $h$ be a maximum area face of $B$.
Project $P$ orthogonally from the direction perpendicular to $h$.
$P$ projects to a convex polygon $X$.
In this projection, $h$ is the minimum area bounding rectangle of $X$,
since otherwise we could rotate the four faces of $B$ that are not perpendicular to $h$
and would get a bounding rectangle smaller than $h$, which in turn would give a bounding box smaller than $B$,
but that would be a contradiction that $B$ is the smallest bounding box.
It implies that the area of $X$ is at least $\frac{1}{2}|h|$.
Now, $C^*$ is at least twice the area of $X$, and $|B|\le 6|h|$. 
Therefore, $C^* \ge 2|X| \ge 2\cdot\frac{1}{2}|h|\ge \frac{1}{6}|B|$.
\qed
\end{proof}

\subsection{Carving phase} 

Let $T=B\backslash P$ be the portion of $B$ that we would achieve if $P$ were removed from $B$.
%is ``trapped'' between the boundaries of $P$ and $B$.
Then, $T$ is a polyhedral object. $T$ may be convex or non-convex and possibly disconnected.
The \emph{outer} surface of $T$ is the surface of $T$ that existed in $B$ when $P$ was not removed from $B$.
Our idea is to apply an edge cut through each edge of $P$, 
and we shall do that in two types of rounds: \emph{face rounds} and \emph{edge rounds}.
% with edge rounds applied inside of face rounds.
Face rounds will find polygonal chains that will partition the faces of $P$  into smaller 
sets and edge rounds will apply edge cuts through the edges of those polygonal chains.
%After an edge cut is applied we shall update $T$ as the portion that stick with $P$.
There will be $O(\log n)$ face rounds,
and within each face round there will be a number of edge rounds but their total cost will be $O(C^*\log n)$. 
Once we have applied edge cuts through all the edges of $P$, each face $f$ of $P$ will have a small ``cap''-like
portion of $T$ over it, which we shall cut at a cost of the area of $f$ to get $P$, 
giving a cost of $O(C^*)$ for all faces.

\subsubsection{Face rounds}
Let $F$ be a set of faces of $P$.
From now on, we use the term \emph{face set} to represent a set of faces of $P$.
%A set $F$ of faces of $P$ is called \emph{connected} if any two vertices in the faces of $F$ have a path.
% and no face other than those in $F$ is completely surrounded by the faces of $F$.
%Let $F$ be a \emph{connected face set} of $l$ faces of $P$.
%We call $F$ as called a \emph{face set} of $P$.
At the very first face round $i=0$, $F$ consists of all the faces of $P$.
We find a chain of edges ${P'}$ that will partition $F$ into two smaller
face sets $F_1$ and $F_2$ by the following lemma.
%(Precise definition of ``connected'' needs detail discussion, which we omit in this extended abstract.)
%But before that, we need some definitions.

\begin{lem}
\label{le:face_equal_partition}
Let $l$ be the number of faces in a face set $F$.
It is always possible to find in $O(l\log l)$ time a non-degenerate orthogonal projection of $P$ 
%which is non-degenerate w.r.t the faces of $F$ 
such that the two sets of visible and invisible faces of $F$ 
%, let them be $F_1$ and $F_2$ respectively,
contain at least $\lfloor\frac{l}{2}\rfloor$ faces each.
\end{lem}

\begin{proof}
For this proof we shall move on to the surface of an origin-centered sphere $s$.
For each face $f\in F$, its  \emph{normal point} is the intersection point of $s$ 
and the outward normal vector of $f$ when the vector is translated to the origin.
%outward normal is uniquely represented by a point of $s$, which we call the \emph{normal point} of $f$.
Each point of $s$ also represents an orthogonal projection direction of $P$.
So, a non-degenerate orthogonal projection of $P$ can be represented 
by a great circle of $s$ that does not pass through the normal points of the faces of $P$.
%Let $N$ be the set of normal points of the faces of $F$.
We need one such great circle satisfying an additional criterion that its two hemispheres contain 
at least $\lfloor\frac{l}{2}\rfloor$ normal points each.
There exists infinitely many such great circles and one of them can be found in $O(l\log l)$ time
as follows. 
Take as \emph{poles} any two antipodal points that are not normal points of the faces of $P$.
Take a great circle $g$ through these two poles and rotate it around these poles until 
the number of normal points in its two hemisphere differ by at most one.
If it happens that some normal points fall on $g$ when we stop,
then slightly change the poles as well as $g$ so that the normal points on $g$ are distributed into two hemispheres as necessary.
For running time, all we need to do is to sort the normal points according to their angular distance 
with the plane of initial position of $g$ at the origin.
The resulting projection is the one from the perpendicular direction of the plane of final position of $g$.
\qed
\end{proof}

The projection direction achieved by the above lemma is called the \emph{zone direction} of $F$.
$P'$ is the chain of edges in the boundary of the above projection
whose corresponding edges in $P$ have both adjacent faces (one is visible and another is invisible) in $F$.
We call $P'$ a \emph{separating chain} of $F$.
$F_1$ and $F_2$ are the two sets of faces separated by $P'$.
In the next face round $i+1$, we shall apply Lemma~\ref{le:face_equal_partition} for each of $F_1$ and $F_2$
recursively and thus get two separating chains and four  face sets.
We shall repeat the same procedure for each of these four face sets.
We shall continue like this until each face set has only one face.
Clearly, we need $O(\log n)$ face rounds.

\subsubsection{Edge rounds}
%We shall apply edge cuts through the edges of ${P'}$ by the edge rounds as described in the next paragraph.
Let ${P'}=e_1,e_2,\ldots,e_k$ be the separating chain of a particular face round
with its two ends, which are two vertices of $e_1$ and $e_k$, touching the outer surface of $T$.
Observe that for the very first face round $i=0$, $P'$ is a cycle and the two ends are the same.
%For any sunsequent edge round, $P'$ is a open polygonal chain.
We shall apply edge cuts through the edges of $P'$ such that all of them are parallel to a particular direction.
Such a direction can be the corresponding zone direction.
We shall call this set of $k$ edge cuts a \emph{zone} of cuts and their direction of cuts the \emph{zone cut direction}.
We shall apply these cuts in $O(\log k)$ edge rounds.
%At each edge round $j$, $0\le j\le \log k$, we shall apply $2^j$ cuts.
At the very first edge round $j=0$, we apply an edge cut through $e_{k/2}$ in the zone cut direction.
%If we are in the very first edge 
This cut will partition the edges of $P'$ into two sub chains of size at most $\lfloor\frac{k}{2}\rfloor$. 
In the next edge round $j+1$, we apply two edge cuts through the two middle edges of these two sub chains,
which will result into four sub chains.
Then in the next round we apply four similar cuts to the four sub chains.
We continue like this until each sub chain has only one edge.
Clearly, we need $O(\log k)$ edge rounds for ${P'}$. 

\begin{lem}
After all the face rounds and their corresponding edge rounds are completed, all edges of $P$ get an edge cut.
\end{lem}

\begin{proof}
Let $e$ be an edge that does not get an edge cut through it. 
Then the two adjacent faces of $e$ are in the same face set.
But that is a contradiction that each face set has only one face.
\qed
\end{proof}

%Therefore, the top ``cap''-like portion of each face $f$ can be cut by a single cut in cost of the area of $f$.

\subsubsection{Analysis}
We are now ready to find the approximation ratio and the running time of our algorithm.

We define the \emph{box area} of a face set $F$ as follows.
When $F$ contains all faces of $P$, its box area is $B$---the whole surface area of $B$.
Zone of cuts through the separating chain of $F$ partitions $F$ into $F_1$ and $F_2$
and $T$ into two components, say $T_1$ and $T_2$, respectively.
Then the \emph{box area} of $F_1$ ($F_2$) is the outer surface area of $T_1$ ($T_2$),
%(i.e., that is not attached to $F_1$ ($F_2$), 
which we denote by by $B_1$ ($B_2$).
Observe that $|B_1|+|B_2|\le |B|$.
Box area of any subsequent face set is similarly defined.
Moreover, two face sets from the same face round have their box areas disjoint
and in any face round sum of all box area is at most $|B|$.

%For any face set $F_m$ and the corresponding box area $B_m$,
%let $d_m$ be the length of maximum diameter of $B_m$.
%Clearly, $|B_m|\ge $ 

The following lemma bounds the cutting cost of an edge rounds in a particular face round.

\begin{lem}
Let ${P}_m'$ be the separating chain with $k$ edges of an arbitrary face set $F_m$ 
to which we apply $O(\log k)$ edge rounds.
Let $B_m$ be the box area of $F_m$.
At each edge round $j$, total cost of $2^j$ cuts is $O(|B_m|)$.
Over all $O(\log k)$ edge rounds, total cost is $O(|B_m|\log n)$.
\end{lem}

\begin{proof}
This proof is similar to that of Lemma~\ref{le:garbage1}.
Consider a particular edge round $j$.
%When $j\ge 1$, there are $2^j$ cuts in this round.
For each cut $c$ the cost of $c$ is no more than the portion of $B_m$ that is thrown away by $c$.
Moreover, these cuts are pairwise disjoint, 
since they can at best intersect the cut which is in between them and was applied in $(j-1)$-th round.
%By Observation~\ref{ob:garbage2}, 
It implies that the total cost of $2^j$ cuts is at most $|B_m|$.
Since $k\le n$, the second part of the lemma follows.
\qed
\end{proof}

The next lemma bounds the total cutting cost over all face rounds.

\begin{lem}
%Let $F$ be the face set consisting of all faces of $P$ to which we apply $O(\log n)$ face rounds.
At each face round $i$, total cost of $2^i$ zones of cuts is $O(|B|\log n)$.
Over all $O(\log n)$ face rounds, the total cost is $O(C^*\log^2n)$.
\end{lem}

\begin{proof}
At each face round $i$, we apply $2^i$ zones of cuts to $2^i$ face sets.
By the previous lemma, for a particular face set $F_m$, $0\le m\le 2^i$,
cost of the zone of cuts applied to it is at most $O(|B_m|\log n)$.
Since $\sum_{1}^{2^{i}} |B_m| \le |B|$, cost of all zone cuts is 
$\sum_{1}^{2^i} O(|B_m|\log n) = O(|B|\log n)$.
Over all $O(\log n)$ face rounds, the total cost is $O(|B|\log^2n)$,
which  by Lemma~\ref{le:lb} is $O(C^*\log^2n)$. 
\qed
\end{proof}

We now see the running time of our algorithm.
Running time in face round $i$ involves finding $2^i$ separating chains,
each of size  $\frac{n}{2^i}$, 
plus applying a zone of cuts to each of them. 
Each separating chain 
%has size at most the number of faces in a face set, which is $\frac{n}{2^i}$ and 
can be found in $O(\frac{n}{2^i}\log \frac{n}{2^i})$ time by Lemma~\ref{le:face_equal_partition}.
%While applying the zone of cuts to a separating chain, 
Each cut needs to update $Q$, which can be done in $O(n)$ time assuming that $Q$ is 
represented by suitable data structures~\cite{berg}.
It gives that a zone of cuts needs $O(\frac{n^2}{2^i})$ time.
So, in round $i$ total time is $ O(2^i(\frac{n^2}{2^i}+\frac{n}{2^i}\log \frac{n}{2^i}))=O(n^2)$.
Over all $O(\log n)$ rounds, it becomes $O(n^2\log n)$.

We summarize the result in the following theorem.

\begin{thm}
Given a convex polyhedron $P$ fixed inside a sphere $Q$, $P$ can be cut out of $Q$ by using only
guillotine cuts in $O(n^3)$ time with cutting cost $O(\log^2n)$ times the optimal,
where $n$ is the number of vertices of $P$.
\end{thm}

\section{Conclusion}
\label{se:con}
In this paper, we have given an $O(n^3)$-time algorithm that cuts a convex polyhedron $P$
with $n$ vertices from a sphere $Q$, where $P$ is fixed inside $Q$, by using guillotine cuts
with cutting cost $O(\log^2 n)$ times the optimal.

This problem is well studied in 2D, where the series of results include several $O(\log n)$ 
and constant factor approximation algorithms and a PTAS.
The key ingredients of the 2D algorithms involve three major steps: 
(1) take some approximate vertex cuts through the vertices of $P$, 
(2) use dynamic programming to find an optimal cutting sequence among the edge cuts,
and the vertex cuts taken in step (1), and 
(3) show that the cutting cost of the sequence obtained in step (2)
is within the desired factor of the optimal.
Using the idea of 2D algorithms may be a way to improve the approximation ratio of our algorithm.
Among the above three steps, it may not be difficult to generalize steps (1) and (3) for 3D,
but the most difficult part we find is the applying a dynamic programming.

An immediate future work would be to find approximation algorithms when $Q$ is another convex polyhedron.
Recently, Ahmed et.al.~\cite{ABHK10} have studied a more generalized version of the problem in 2D
where the polygon $P$ is \emph{not fixed} inside a circle $Q$.
For that problem they have given several constant factor approximation algorithms. 
It would be interesting to study that version of the problem in 3D.

%\paragraph{Acknowledgment}
%We acknowledge the support of Bangladesh University of Engineering and Technology.

%\renewcommand{\baselinestretch}{.9}

\bibliographystyle{abbrv}

\bibliography{bib2_short}

\end{document}